\documentclass{llncs}

\usepackage{algorithm}
\usepackage{algpseudocode}
\usepackage{graphicx}
\usepackage{color}
\usepackage{transparent}

\begin{document}

\title{Combining All Pairs Shortest Paths and All Pairs Bottleneck Paths Problems\thanks{This research was supported by the EU/NZ Joint Project, Optimization and its Applications in Learning and Industry (OptALI).}}

\author{Tong-Wook Shinn \and Tadao Takaoka}
\institute{Department of Computer Science and Software Engineering\\
University of Canterbury\\
Christchurch, New Zealand}

\maketitle

\begin{abstract}
We introduce a new problem that combines the well known All Pairs Shortest Paths (APSP) problem and the All Pairs Bottleneck Paths (APBP) problem to compute the shortest paths for all pairs of vertices for all possible flow amounts. We call this new problem the All Pairs Shortest Paths for All Flows (APSP-AF) problem. We firstly solve the APSP-AF problem on directed graphs with unit edge costs and real edge capacities in $\tilde{O}(\sqrt{t}n^{(\omega+9)/4}) = \tilde{O}(\sqrt{t}n^{2.843})$ time, where $n$ is the number of vertices, $t$ is the number of distinct edge capacities (flow amounts) and $O(n^{\omega}) < O(n^{2.373})$ is the time taken to multiply two $n$-by-$n$ matrices over a ring. Secondly we extend the problem to graphs with positive integer edge costs and present an algorithm with $\tilde{O}(\sqrt{t}c^{(\omega+5)/4}n^{(\omega+9)/4}) = \tilde{O}(\sqrt{t}c^{1.843}n^{2.843})$ worst case time complexity, where $c$ is the upper bound on edge costs.
\end{abstract}

\section{Introduction}
\label{sec:intro}
Finding the shortest paths between pairs of vertices in a graph is one of the most extensively studied problems in algorithms research. The shortest paths problem is often categorized into the Single Source Shortest Paths (SSSP) problem, which is to compute the shortest paths between one source vertex to all other vertices in the graph, and the All Pairs Shortest Paths (APSP) problem, which is to compute the shortest paths between all possible pairs of vertices on the graph.

Arguably the most famous algorithm for the APSP problem is Floyd's algorithm that runs in $O(n^{3})$ time. There have been many attempts at providing sub-cubic time bounds for solving the APSP problem on dense graphs with real edge costs \cite{Fredman,Dobosiewicz,Tad04,Zwick06,Han,Chan,HT}, all achieving time improvements by logarithmic factors. The current best time bound is $O(n^{3}\log{\log{n}}/\log^{2}{n})$ by Han and Takaoka \cite{HT}. If the graph has integer edge costs, faster matrix multiplication over a ring \cite{Strassen} can be utilized to achieve deeply sub-cubic time bounds. Alon, Galil and Margalit achieved $O(n^{(3+\omega)/2})$ time bound for solving the APSP problem on directed unweighted graphs, where $O(n^{\omega})$ is the time bound on multiplying two $n$-by-$n$ matrices over a ring \cite{AGM}. This time complexity translates to $O(n^{2.687})$ with $\omega < 2.373$ \cite{Williams}. The best time bound for this problem is currently $O(n^{2.530})$ by Zwick \cite{Zwick02}, thanks to Le Gall's recent achievement in rectangular matrix multiplication \cite{Gall}.

Another well studied problem in graph theory is finding the maximum bottleneck between pairs of vertices. The bottleneck of a path is the minimum capacity of all edges on the path. The problem of finding the paths that give the maximum bottlenecks for all pairs of vertices is formally known as the All Pairs Bottleneck Paths (APBP) problem. Vassilevska, Williams and Yuster achieved $O(n^{2+\omega/3}) = O(n^{2.791})$ time bound for solving the APBP problem on graphs with real edge capacities \cite{VWY}, and this has subsequently been improved to $O(n^{(\omega+3)/2}) = O(n^{2.687})$ by Duan and Pettie \cite{DP}.

Let us consider a path that gives us the bottleneck value of $b$ from vertex $u$ to vertex $v$. In other words, we can push flows of amounts up to $b$ from $u$ to $v$ using this path. If the flow demand from $u$ to $v$ is less than $b$, however, there may be a shorter path. This information is useful if we wish to minimize the path cost (distance) for varying flow demands. Thus we combine the two well known APSP and APBP problems and compute the shortest paths for all pairs for all possible flow demands. We call this new problem the All Pairs Shortest Paths for All Flows (APSP-AF) problem. Note that this is different from the All Pairs Bottleneck Shortest Paths (APBSP) problem \cite{VWY}, which is to compute the bottlenecks of the shortest paths for all pairs. There are obvious practical applications for the APSP-AF problem in any form of network analysis, such as computer networks, transportation and logistics, etc.

In this paper we present two algorithms for solving the APSP-AF problem on directed graphs with positive integer edge costs and real edge capacities. Firstly we present an algorithm to solve the problem on graphs with unit edge costs in $O(\sqrt{t}n^{(\omega+9)/4}) = O(\sqrt{t}n^{2.843})$ time, where $t$ is the number of distinct edge capacities. We then extend this algorithm to solve the problem on graphs with positive integer edge costs of at most $c$ in $O(\sqrt{t}c^{(\omega+5)/4}n^{(\omega+9)/4}) = O(\sqrt{t}c^{1.843}n^{2.843})$ time, which is reduced to the complexity of the first algorithm when $c=1$.

\section{Preliminaries}
\label{sec:prem}
Let $G=\{V,E\}$ be a directed graph with non-negative integer edge costs and real edge capacities. Let $n = |V|$ and $m = |E|$. Vertices (or nodes) are given by integers such that $\{1,2,3,...,n\} \in V$. Let $(i,j)$ denote the edge from vertex $i$ to vertex $j$. Let $cost(i,j)$ denote the cost and $cap(i,j)$ denote the capacity of the edge $(i,j)$. Let $t$ be the number of distinct $cap(i,j)$, and let $c$ be the upper bound on $cap(i,j)$. We define path \emph{length} as the number of edges on the path, irrespective of their costs or capacities. We define path \emph{cost} or \emph{distance} as the sum of all edge costs on the path.

We represent $G$ in a series of matrices. Let $R^{\ell} = \{r^{\ell}_{ij}\}$ be the reachability matrix, for $0 < \ell < n$, where $r^{\ell}_{ij} = 1$ if $j$ is reachable from $i$ via some path of length up to $\ell$ and $r^{\ell}_{ij} = 0$ otherwise. $r^{\ell}_{ii} = 1$ for all $\ell$. $r^{1}_{ij} = 1$ if an edge exists from $i$ to $j$, and 0 otherwise. $R^{1}$ is called the adjacency matrix of $G$. Let $C^{\ell}=\{c^{\ell}_{ij}\}$ be the capacity matrix, where $c^{\ell}_{ij}$ represents the maximum possible capacity (or bottleneck) from $i$ to $j$ via any paths of lengths up to $\ell$. $c^{\ell}_{ii} = \infty$ for all $\ell$. $c^{1}_{ij} = cap(i,j)$ if there is an edge from $i$ to $j$, and 0 otherwise. Let $D^{\ell}=\{d^{\ell}_{ij}\}$ be the distance matrix, where $d^{\ell}_{ij}$ represents the shortest possible distance from $i$ to $j$ via any paths of lengths up to $\ell$. $d^{\ell}_{ii} = 0$ for all $\ell$. $d^{1}_{ij} = cost(i,j)$ if there is an edge from $i$ to $j$, and $\infty$ otherwise.

Let $X \ast Y$ denote the $(min,+)$-product and $X \star Y$ denote the $(max,min)$-product of the two matrices $X$ and $Y$, where:
$$
X \ast Y = \min\limits_{k=1}^{n}\{x_{ik}+y_{kj}\} \quad X \star Y = \max\limits_{k=1}^{n}\{\min\{x_{ik}, y_{kj}\}\}
$$
\noindent Clearly the $(min,+)$-product is applicable to the distance matrix whereas the $(max,min)$-product is applicable to the capacity matrix.

\section{Review of the algorithm by Alon, Galil and Margalit}
\label{sec:agm}

Our algorithm for solving the APSP-AF problem is largely based on the algorithm given by Alon \emph{et al.} \cite{AGM}. Therefore we provide a review of this algorithm using the same set of terminologies as an earlier review of the same algorithm by Takaoka \cite{Tad95}. The algorithm under review computes the All Pairs Shortest Distances (APSD) on directed graphs with unit edge costs. In summary this algorithm achieves sub-cubic time bound by utilizing faster matrix multiplication over a ring to perform Boolean matrix multiplication, and also using the novel idea of \emph{Bridging Sets}.

\begin{algorithm}
\caption{Algorithm by Alon, Galil and Margalit}
\label{alg:agm}
\begin{algorithmic}
\algnotext{EndFor}
\algnotext{EndIf}
\algnotext{EndWhile}
\newcommand{\IfLine}[2]{%
    \State\algorithmicif\ {#1}\ \algorithmicthen\ {#2}%
}
\Statex{/* Acceleration Phase*/}
\For{$\ell = 2$ to $r$}
	\State{$R^{\ell} \leftarrow R^{\ell - 1} \times R^{1}$ /* Boolean matrix multiplication */}
	\For{$i \leftarrow 1$ to $n$; $j \leftarrow 1$ to $n$}
		\IfLine{$r^{\ell}_{ij} = 1$ and $d^{\ell - 1}_{ij} = \infty$}{$d^{\ell}_{ij} \leftarrow \ell$}
		\IfLine{$d^{\ell - 1}_{ij} < \ell$}{$d^{\ell}_{ij} \leftarrow d^{\ell - 1}_{ij}$}
	\EndFor
\EndFor

\mbox{}

\Statex{/* Cruising Phase*/}
\While{$\ell < n$}
	\State{$\ell' \leftarrow \lceil \frac{3\ell}{2} \rceil$}
	\For{$i \leftarrow 1$ to $n$}
		\State{Scan $i^{th}$ row of $D^{\ell}$ with $j$ and find the smallest set of equal $d^{\ell}_{ij}$ such that}
		\State{\mbox{ } \mbox{ } $\lceil \ell/2 \rceil \leq d^{\ell}_{ij} \leq \ell$ and let the set of corresponding $j$ be $S_{i}$}
	\EndFor
	\For{$i \leftarrow 1$ to $n$; $j \leftarrow 1$ to $n$}
		\State{$m_{ij} \leftarrow \min_{k \in S_{i}}\{d^{\ell}_{ik} + d^{\ell}_{kj}\}$ /* Squaring $D^{\ell}$ with $S_{i}$ */}
		\If{$d^{\ell}_{ij} \leq \ell$}
			\State{$d^{\ell'}_{ij} \leftarrow d^{\ell}_{ij}$}
		\ElsIf{$m_{ij} \leq \ell'$}
			\State{$d^{\ell'}_{ij} \leftarrow m_{ij}$}
		\EndIf
	\EndFor
	\State{$\ell \leftarrow \ell'$}
\EndWhile
\end{algorithmic}
\end{algorithm}

Algorithm \ref{alg:agm} consists of two phases. We refer to the first part of the algorithm as the \emph{acceleration phase}, and the second part of the algorithm as the \emph{cruising phase}. The acceleration phase repeatedly performs Boolean matrix multiplication with the adjacency matrix to compute APSD for all pairs with distances up to $\ell = r$, where $r$ is a constant such that $1 < r < n$. Clearly this only works on graphs with unit edge costs where the path length and the path cost are equivalent. The algorithm then switches to the cruising phase where the ordinary multiplication method is used with the help of bridging sets, $S_{i}$, where $S_{i}$ is a set of ``via'' vertices for all rows $i$ of the distance matrix $D$. That is, when computing $d^{\ell}_{ik} + d^{\ell}_{kj}$ for the $(min,+)$-product, we inspect only the set of vertices in $S_{i}$ for $k$ rather than inspecting all $O(n)$ elements. Alon \emph{et al.} have shown that with path lengths equal to $r$, the size of the bridging set $S_{i}$ for each row $i$ is bounded by $O(n/r)$ \cite{AGM}. Hence we start the cruising phase with $|S_{i}| = O(n/r)$ for each row $i$.

The acceleration phase takes $O(rn^\omega)$ time, and the cruising phase performs repeated squaring of the distance matrix in $O(n^{2}\cdot{}\frac{n}{r})$ time. Alon \emph{et al.} chose to increase the path length by a factor of $\frac{3}{2}$ in each iteration of the cruising phase. This factor of $\frac{3}{2}$ is somewhat arbitrary, as any factor greater than $1$ and less than $2$ can be used. Because the size of the bridging set decreases by a constant factor in each iteration, we end up with a geometric series if we add up the time complexities of each iteration, and hence the first squaring dominates the time complexity. The total time complexity of $O(n^{(3 + \omega)/2}) = O(n^{2.687})$ of this algorithm comes from balancing the time complexities of the two phases to retrieve the best value for $r$, that is, setting $rn^\omega = n^{2}\cdot{}\frac{n}{r}$ then solving for $r$.

\section{APSP-AF on graphs with unit edge costs}
\label{sec:apspaf}
We first consider solving the All Pairs Shortest Distances for All Flows (APSD-AF) problem on directed graphs with unit edge costs, that is, computing only the shortest distances rather than actual path. Path lengths and path distances are used interchangeably in this section. To re-iterate the APSD-AF problem, for each pair of vertices $(i,j)$ for each possible flow amount, we want to compute the shortest distance. Thus our aim here is to obtain a set of $(d,f)$ pairs for all pairs of vertices, where $f$ is the maximum flow amount that can be pushed through the shortest path whose length (distance) is $d$. We refer to the distinct capacity values as \emph{maximal flows}. i.e. there are $t$ maximal flows. Assume that the maximal flows are sorted in increasing order. If we wish to push $f$ such that $f_{1} < f < f_{2}$ for consecutive maximal flows $f_{1}$ and $f_{2}$, then clearly $f$ is represented by $f_{2}$.

\begin{figure}
\def\svgwidth{250pt}
\begin{center}
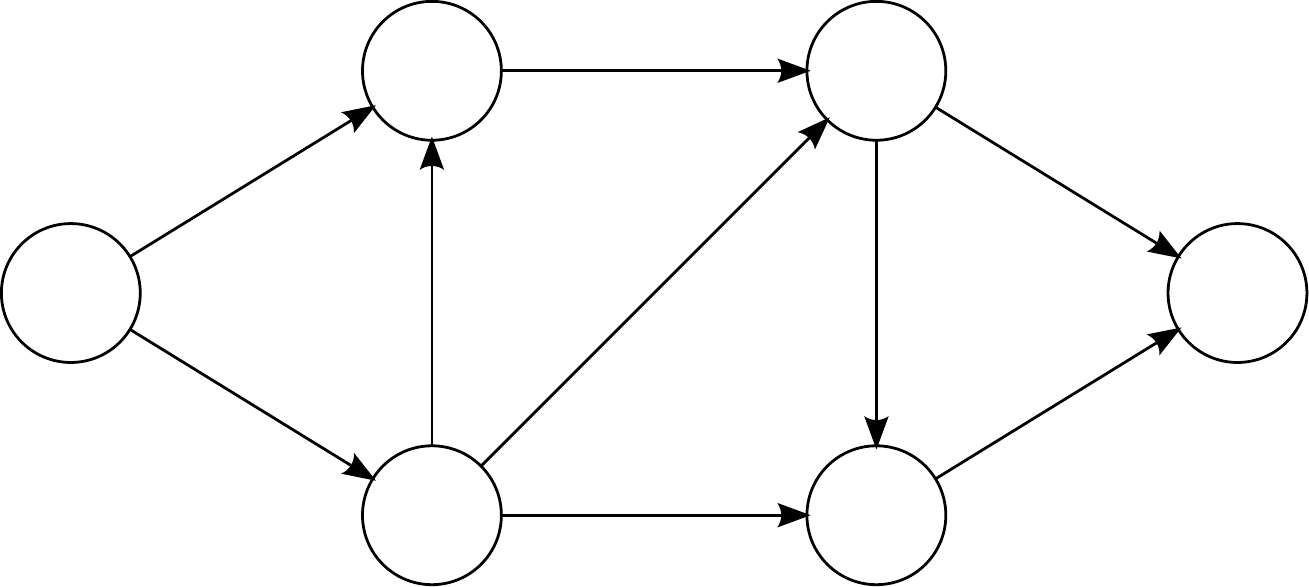
\end{center}
\caption{An example graph with $n=6$, $m=9$, $t=7$ and $c=3$. Numbers in the parenthesis beside each edge shows the edge cost and capacity, respectively.}
\label{fig:graph}
\end{figure}


Let $U$ be a matrix such that $u_{ij}$ is a set of $(d,f)$ pairs as described above. Let both $(d,f)$ and $(d',f')$ be in $u_{ij}$ such that $d < d'$. We keep $(d',f')$ \emph{iff} $f < f'$. In other words, a longer path is only useful to us if it can accommodate a greater flow. If $d = d'$, we keep the pair that provides the bigger flow. Since there can only be $n-1$ different values of $d$, each $u_{ij}$ has at most $n-1$ pairs of $(d, f)$. We assume the pairs are sorted in ascending order of $d$. We make an interesting observation here that once all $(d,f)$ pairs for all $u_{ij}$ are computed (i.e. the APSP-AF problem is solved), the first pairs for all $u_{ij}$ is the solution to the APBSP problem, and the last pairs for all $u_{ij}$ is the solution to the APBP problem.

\begin{example}
If the graph in Figure \ref{fig:graph} had unit edge costs instead of the varying integer edge costs, solving APSD-AF on the graph would result in three $(d,f)$ pairs from vertex 1 to vertex 6, that is, $u_{(1)(6)} = \{(3,5), (4,6), (5,7)\}$.
\end{example}

We now introduce Algorithm \ref{alg:apsdaf} to solve the APSD-AF problem on directed graphs with unit edge costs. Let $P^{f}$ be the approximate distance matrix for shortest paths that can accommodate flows up to $f$. In the acceleration phase, we compute the maximum bottleneck values for all possible path lengths up to $r$ for all pairs, where $r$ is a constant such that $1 < r < n$. Then from the results gathered in the acceleration phase, we prepare a series of distance matrices, $P^{f}$, one for each maximal flow value $f$, and move onto the cruising phase where we compute the shortest distances for all pairs for all flows by repeatedly squaring each $P^{f}$.

\begin{algorithm}
\caption{Solve APSD-AF on graphs with unit edge costs}
\label{alg:apsdaf}
\begin{algorithmic}
\algnotext{EndFor}
\algnotext{EndIf}
\algnotext{EndWhile}
\Statex{/* Initialization for acceleration phase */}
\For{$i \leftarrow 1$ to $n$; $j \leftarrow 1$ to $n$}
	\State{$u_{ij} \leftarrow \phi$} /* $\phi$ is empty */
\EndFor

\mbox{}

\Statex{/* Acceleration phase */}
\For{$\ell \leftarrow 2$ to $r$}
	\State{$C^{\ell} \leftarrow C^{\ell-1} \star C^{1}$ /* $(max,min)$ matrix multiplication */}
	\For{$i \leftarrow 1$ to $n$; $j \leftarrow 1$ to $n$; $i \neq j$}
		\If{$c^{\ell}_{ij} > c^{\ell-1}_{ij}$}
			\State{Append $(\ell,c^{\ell}_{ij})$ to $u_{ij}$}
		\EndIf
	\EndFor
\EndFor

\mbox{}

\Statex{/* Initialization for cruising phase */}
\State{$P^{f} \leftarrow I$ for all maximal flow $f$ /* $I$ has 0 diagonal elements and $\infty$ for others */}
\For{$i \leftarrow 1$ to $n$; $j \leftarrow 1$ to $n$; $i \neq j$}
	\State{Let $u_{ij} = \{(d_{1},f_{1}), (d_{2},f_{2}), ..., (d_{s},f_{s})\}$ for some $s$ /* We skip empty $u_{ij}$ */}
	\State{$k \leftarrow 1$ /* $k$ iterates from $1$ to $s$ */}
	\ForAll{maximal flow $f$ in increasing order}
		\If{$f > f_{k}$}
			\State{$k \leftarrow k + 1$ /* The next $d_{k}$ value is needed */}
		\EndIf
		\If{$k > s$}
			\State{break /* We proceed to the next $u_{ij}$ */}
		\EndIf
		\State{$p^{f}_{ij} \leftarrow d_{k}$}
	\EndFor
\EndFor

\mbox{}

\Statex{/* Cruising phase */}
\ForAll{maximal flow $f$}
	\State{Perform cruising phase of Algorithm \ref{alg:agm} on $P^{f}$}
\EndFor

\mbox{}

\Statex{/* Finalization */}
\For{$i \leftarrow 1$ to $n$; $j \leftarrow 1$ to $n$; $i \neq j$}
	\ForAll{maximal flow $f$ in increasing order}
		\State{$d \leftarrow p^{f}_{ij}$}
		\State{Let the last pair of $u_{ij}$ be $x=(d',f')$ /* If $u_{ij}$ is empty, $x = \phi$ */}
		\If{$x = \phi$ or $(f > f'$ and $d < \infty)$}
			\If{$d = d'$ /* This condition is false if $x = \phi$ */}
				\State{Replace $x$ with $(d,f)$}
			\Else
				\State{Append $(d,f)$ to $u_{ij}$}
			\EndIf
		\EndIf
	\EndFor
\EndFor
\end{algorithmic}
\end{algorithm}

\begin{lemma}
Algorithm \ref{alg:apsdaf} correctly solves APSD-AF on directed graphs with unit edge costs.
\end{lemma}
\begin{proof}
In the acceleration phase, instead of performing Boolean matrix multiplication as in Algorithm \ref{alg:agm}, we compute the $(max,min)$-product with the capacity matrices $C^{1}$ and $C^{\ell-1}$. After each matrix multiplication, if a path of greater capacity has been found for the vertex pair $(i,j)$, we append the pair $(\ell, c^{\ell}_{ij})$ to $u_{ij}$ since we have found a longer path that can accommodate a greater flow. Thus after the $r^{th}$ iteration of the acceleration phase, all relevant $(d,f)$ pairs for all $u_{ij}$ are found such that $d \leq r$.

After the acceleration phase we initialize the approximate distance matrices $P^{f}$ from $U$, one matrix for each maximal flow $f$, in preparation for the cruising phase. Note that if the $(d,f)$ pair for a given flow value $f$ does not exist in $u_{ij}$, we take the next pair $(d',f')$ in $u_{ij}$ (if one exists) and let $p^{f}_{ij} = d'$.

At this stage, if $p^{f}_{ij} < \infty$, $p^{f}_{ij}$ is already the length of the shortest path from $i$ to $j$ that can push flow $f$. Thus the actual aim of the cruising phase of this algorithm is to compute the shortest distance for all other elements in $P^{f}$ such that $p^{f}_{ij} = \infty$ at the start of the cruising phase. Note that unless $G$ is strongly connected, some elements of $P^{f}$ will remain at $\infty$ until the end of the algorithm. The aim of the cruising phase is achieved by repeatedly squaring each $P^{f}$ with the help of the bridging set, as proven in \cite{AGM}.

Retrieving sets of $(d,f)$ pairs after the cruising phase from each resulting $P^{f}$ is simply a reverse process of the initialization for the cruising phase, and thus our search for all sets of $(d,f)$ pairs for all $(i,j)$ is complete after finalization.\qed
\end{proof}

\begin{lemma}
Algorithm \ref{alg:apsdaf} runs in $O(\sqrt{t}n^{(\omega+9)/4}) = O(\sqrt{t}n^{2.843})$ worst case time.
\end{lemma}
\begin{proof}
For the acceleration phase we use the the current best known algorithm to compute the $(max,min)$-product in each iteration, which gives us the time bound of $O(rn^{(3+\omega{})/2})$ \cite{DP}. The time complexity for the cruising phase is $O(tn^3/r)$ since there are a total of $t$ maximal flows, each taking $O(n^3/r)$ time to finish the computation of APSD. The time bound for the initialization for the cruising phase and the finalization is $O(tn^2)$, which is absorbed by $O(tn^3/r)$ since $n/r > 1$. We balance the time complexities of the acceleration phase and the cruising phase by setting $r = \sqrt{t}n^{(3-\omega)/4}$, and this gives us the total worst case time complexity of $O(\sqrt{t}n^{(\omega+9)/4})$.\qed
\end{proof}

If $t = O(n^{2})$, the value we choose for $r$ may exceed $n$. In such a case, we simply stay in the acceleration phase until $r = n - 1$. Thus a more accurate worst case time complexity of Algorithm \ref{alg:apsdaf} is actually $O(\min{\{n^{(5+\omega)/2}, \sqrt{t}n^{(\omega+9)/4}\}})$.

A straightforward method of solving the APSD-AF problem is to repeatedly compute APSD for each maximal flow value $f$ using only edges that have capacities greater than or equal to $f$. This method is equivalent to starting the cruising phase at $r=1$, giving us the time complexity of $O(tn^{2.530})$ if we use Zwick's algorithm to solve APSD for each maximal flow value \cite{Zwick02}. For $t > n^{0.626}$, Algorithm \ref{alg:apsdaf} is faster. Note that a simple decremental algorithm where edges are removed in the reverse order of capacities while repeatedly solving APSD cannot be used to solve the APSD-AF problem because edges with larger capacities may later be required to provide shorter paths for a smaller maximal flow values.

\begin{theorem}
\label{the:apsdaf}
There exists an algorithm that can solve APSP-AF on directed graphs with uni edge costs in $\tilde{O}(\sqrt{t}n^{{(\omega+9)/4}})$ worst case time.
\end{theorem}
\begin{proof}
As noted earlier there can be $O(n)$ $(d,f)$ pairs for each vertex pair $(i,j)$. Since the lengths of each path can be $O(n)$, explicitly listing all paths takes $O(n^4)$ time. We get around this by modifying Algorithm \ref{alg:apsdaf} to extend the $(d,f)$ pair to the $(d,f,s)$ triplet, where $s$ is the \emph{successor} node, such that retrieving the actual path from $(d,f,s)$ can be performed by simply following the successor nodes. In the acceleration phase witnesses for the $(max,min)$-product can be retrieved with an extra polylog factor \cite{DP}, and the successor nodes can be computed from the witnesses in each iteration in $O(n^2)$ time \cite{Zwick02}. In the cruising phase retrieving the witnesses, and hence the successor nodes, is a simple exercise since ordinary matrix multiplication is performed. Therefore extending $(d,f)$ to $(d,f,s)$ only takes an additional polylog factor.

The explicit path for a given flow demand from $i$ to $j$ can be generated in time linear to the path length as follows. Firstly we perform binary search for the triplet $(d,f,s)$ in $u_{ij}$ with $f$ as the key to find the minimum distance $d$ such that $f$ is greater than or equal to the given flow requirement. We then traverse the successor nodes $s$ one by one, using $d$ to look up each subsequent successor node in $O(1)$ time.\qed
\end{proof}

\section{APSP-AF on graphs with integer edge costs}
\label{sec:apspaf2}

We now consider solving the APSD-AF problem on directed graphs with integer edge costs and real edge capacities, where the edge cost is bounded by $c$. Note that with integer edge costs we need to make a clear distinction between path lengths and distances. One approach for solving this problem is to use the method described in \cite{AGM} to replace $G$ with an expanded graph $G'$ such that all edges in $G'$ have unit edge costs, then applying the algorithm on $G'$ to solve the problem on $G$. $G'$ is created by attaching a chain of $c-1$ artificial vertices to each real vertex such that the artificial edges linking the artificial vertices in each chain have unit edge costs and capacities of $\infty$. We then replace each real edge $(i,j)$ with an artificial edge with unit edge cost and capacity of $cap(i,j)$ by choosing one of the artificial vertices of $i$ (or $i$ itself) as the source vertex and the real vertex $j$ as the destination, such that there exists a path from $i$ to $j$ with length equal to $cost(i,j)$. See Figure \ref{fig:g_prime} for an illustration of how a graph is expanded. The expanded graph $G'$ has $O(cn)$ vertices, and we can clearly solve APSD-AF on $G$ by solving APSD-AF on $G'$ in $O(\sqrt{t}(cn)^{(9+\omega)/4})$ time.


\begin{example}
Solving APSD-AF on the graph in Figure \ref{fig:graph} results in a total of five $(d,f)$ pairs from vertex 1 to 6, that is, $u_{(1)(6)} = \{(4,2), (6,3), (7,5), (8,6), (9,7)\}$.
\end{example}

We can do better, however, with the key observation that only the acceleration phase of Algorithm \ref{alg:apsdaf} is restricted to graphs with unit edge costs. In other words, we can complete the acceleration phase on the expanded graph $G'$, gather the intermediate results, and then finish off the remaining computation after contracting the graph back to $G$. We need care here, as the path lengths in $G'$ are actually equivalent to the path costs in $G$, and the bridging sets in the cruising phase are determined from the path lengths rather than the path costs. Therefore we need to make substantial changes to Algorithm \ref{alg:apsdaf} to keep track of both the path lengths and the path costs of $G$ in the acceleration phase, as well as modifying the cruising phase to correctly use the path lengths of $G$ in determining the bridging sets.

\begin{figure}
\def\svgwidth{250pt}
\begin{center}
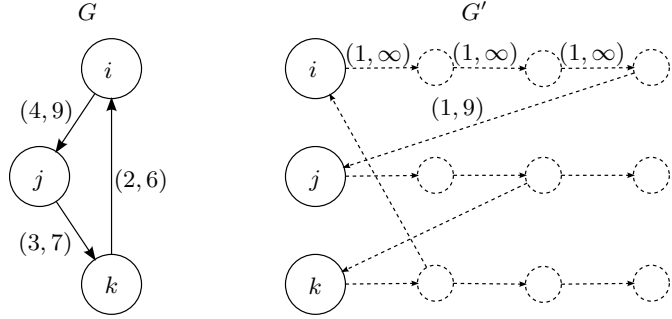
\end{center}
\caption{Expanding $G$ to $G'$ with $c = 4$.}
\label{fig:g_prime}
\end{figure}

Firstly we extend the pair $(d,f)$ to the triplet $(h,d,f)$, where $h$ is the path length in $G$, $d$ is the path cost in $G$ (i.e. the path length in $G'$) and $f$ is the maximal flow. We introduce $U' = \{u'_{ij}\}$ where $u'_{ij}$ is a set of triplets $(h,d,f)$ for all pairs of vertices in $G'$. We omit the superscript $\ell$ that denotes the path length in the following matrix definitions. Let $C' = \{c'_{ij}\}$ be the capacity matrix of $G'$ and let $W = \{w_{ij}\}$ be the witness matrix for the $(max,min)$-product. Let $Q^{f} = \{q^{f}_{ij}\}$ such that $q^{f}_{ij}$ is the length of the path that gives the path cost (distance) of $p^{f}_{ij}$, where $P^{f} = \{p^{f}_{ij}\}$ is the distance matrix as defined in Section \ref{sec:apspaf}. That is, $p^{f}_{ij}$ is the minimum path cost (distance) of all paths from $i$ to $j$ that can push flow of amount $f$. Note that $h$ in the triplet $(h,d,f)$ is no longer required once all $Q^{f}$ are initialized before the start of the cruising phase.

\begin{algorithm}
\caption{Solve APSD-AF on directed graphs with non-negative integer edge costs}
\label{alg:apsdaf2}
\begin{algorithmic}
\algnotext{EndFor}
\algnotext{EndIf}
\algnotext{EndWhile}
\newcommand{\IfLine}[2]{%
    \State\algorithmicif\ {#1}\ \algorithmicthen\ {#2}%
}
\Statex{/* Initialization for acceleration phase */}
\State{Create $G'$ from $G$ /* $G$ is expanded to $G'$ */}
\For{$i \leftarrow 1$ to $cn$; $j \leftarrow 1$ to $cn$}
	\State{$u'_{ij} \leftarrow \phi$}
\EndFor

\mbox{}

\Statex{/* Acceleration phase */}
\For{$\ell \leftarrow 2$ to $r$}
	\State{$C'^{(\ell)} \leftarrow C'^{(\ell-1)} \star C'^{1}$ /* Witnesses given as $W = \{w_{ij}\}$ */}
	\For{$i \leftarrow 1$ to $cn$; $j \leftarrow 1$ to $cn$; $i \neq j$}
		\If{$c'^{(\ell)}_{ij} > c'^{(\ell-1)}_{ij}$}
			\State{Let $k = w_{ij}$, and $(h,d,f) \leftarrow$ last triplet in $u'_{ik}$ /* If empty, h = 0 */}
			\If{$j \in G$ /* If $j$ is a real vertex */}
				\State{Append $(h+1,\ell,c'^{(\ell)}_{ij})$ to $u'_{ij}$}
			\Else
				\State{Append $(h,\ell,c'^{(\ell)}_{ij})$ to $u'_{ij}$}
			\EndIf
		\EndIf
	\EndFor
\EndFor

\mbox{}

\Statex{/* Initialization for cruising phase */}
\State{$U \leftarrow $rows and columns in $U'$ for real vertices /* $G'$ is contracted back to $G$ */}
\State{$P^{f,\ell}, Q^{f,\ell} \leftarrow I$ for all maximal flow $f$}
\For{$i \leftarrow 1$ to $n$; $j \leftarrow 1$ to $n$; $i \neq j$}\label{line:ic_start}
	\State{Let $u_{ij} = \{(h_{1},d_{1},f_{1}), ..., (h_{s},d_{s},f_{s})\}$ for some $s$ /* Skip empty $u_{ij}$ */}
	\State{$k \leftarrow 1$ /* $k$ iterates from $1$ to $s$ */}
	\ForAll{maximal flow $f$ in increasing order}
		\IfLine{$f > f_{k}$}{$k \leftarrow k + 1$ /* The next $d_{k}$ value is needed */}
		\IfLine{$k > s$}{break /* We proceed to the next $u_{ij}$ */}
		\State{$p^{f,\ell}_{ij} \leftarrow d_{k}$; $q^{f,\ell}_{ij} \leftarrow h_{k}$}
	\EndFor
\EndFor\label{line:ic_end}

\mbox{}

\Statex{/* Cruising phase */}
\ForAll{maximal flow $f$}
	\While{$\ell < n$}
		\State{$\ell' \leftarrow \lceil \frac{3\ell}{2} \rceil$}
		\For{$i \leftarrow 1$ to $n$}
			\State{Scan $i^{th}$ row of $Q^{f,\ell}$ with $j$ to find the smallest set of equal $q^{f,\ell}_{ij}$ such that}
			\State{\mbox{ } \mbox{ } $\lceil \ell/2 \rceil \leq q^{f,\ell}_{ij} \leq \ell$ and let the set of corresponding $j$ be $S_{i}$}
		\EndFor
		\For{$i \leftarrow 1$ to $n$; $j \leftarrow 1$ to $n$}
			\State{$m_{ij} \leftarrow \min_{k \in S_{i}}\{p^{f,\ell}_{ik} + p^{f,\ell}_{kj}\}$}
			\State{$k \leftarrow$ the vertex that gives above $m_{ij}$ such that $q^{f,\ell}_{ik} + q^{f,\ell}_{kj}$ is minimum}
			\If{$m_{ij} < p^{f,\ell}_{ij}$}
				\State{$p^{f,\ell'}_{ij} \leftarrow m_{ij}; q^{f,\ell'}_{ij} \leftarrow q^{f,\ell}_{ik} + q^{f,\ell}_{kj}$}
			\Else
				\State{$p^{f,\ell'}_{ij} \leftarrow p^{f,\ell}_{ij}; q^{f,\ell'}_{ij} \leftarrow q^{f,\ell}_{ij}$}
			\EndIf
		\EndFor
		\State{$\ell \leftarrow \ell'$}
	\EndWhile
\EndFor
	
\mbox{}

\Statex{/* Finalization - same as Algorithm \ref{alg:apsdaf} */}
\end{algorithmic}
\end{algorithm}

\begin{lemma}
Algorithm \ref{alg:apsdaf2} correctly solves APSD-AF on directed graphs with non-negative integer edge costs.
\end{lemma}
\begin{proof}
We start by creating $G'$ from $G$ then proceed to the acceleration phase. We only need to show that the actual path length $h$ in the triplet $(h,d,f)$ is correctly determined, since we have already discussed the $(d,f)$ pairs in Section \ref{sec:apspaf}. What is effectively happening in the acceleration phase of Algorithm \ref{alg:apsdaf2} is that the path length information is carried from one real vertex to the next real vertex by the artificial vertices in between. Since we are multiplying by $C'^{(1)}$ in each iteration, the witness $w_{ij}$ will always be the vertex that comes straight before the destination vertex $j$ on the path from $i$ to $j$. That is, it is not possible for any vertices (real or artificial) to exist between $k = w_{ij}$ and $j$. Therefore we retrieve the last $(h,d,f)$ triplet from $u'_{ik}$ and increment the given $h$ \emph{iff} $j$ is a real vertex. Thus the correct path length in $G$ is given by $h$ at the end of the acceleration phase since we are not counting the artificial vertices in the path.

The changes made to the cruising phase is to ensure that the bridging sets $S_{i}$ is correctly determined from the path lengths rather than the path costs. Note that in Algorithm \ref{alg:apsdaf} this distinction was unnecessary because we were only considering graphs with unit edge costs. Clearly the correctness of the crusing phase remains intact by keeping $q^{f,\ell}_{ij}$ updated alongside $p^{f,\ell}_{ij}$.\qed
\end{proof}

\begin{lemma}
Algorithm \ref{alg:apsdaf2} runs in $\tilde{O}(\sqrt{t}c^{(\omega+5)/4}n^{(\omega+9)/4}) = \tilde{O}(\sqrt{t}c^{1.843}n^{2.843})$ worst case time.
\end{lemma}
\begin{proof}
The time complexity of the acceleration phase is $\tilde{O}(r(cn)^{(3+\omega)/2})$ since there are $O(cn)$ vertices in $G'$. After the $r^{th}$ iteration in the acceleration phase, we have computed the bottleneck for all paths of lengths up to $r$, but this is path lengths in the expanded graph $G'$, and not $G$. We divide $r$ by $c$ to retrieve the lower bound on the path lengths in the original graph $G$ after the acceleration phase. Therefore the time complexity of the cruising phase is $O(tcn^{3}/r)$. Both the time complexities for initialization for cruising phase and finalization are again absorbed by the time complexity of the cruising phase. We balance the time complexities of the acceleration phase and the cruising phase by setting $r = \sqrt{t}c^{(-1-\omega)/4}n^{(3-\omega)/4}$, which gives us the total worst case time complexity of $\tilde{O}(\sqrt{t}c^{(\omega+5)/4}n^{(\omega+9)/4})$.\qed
\end{proof}

\begin{theorem}
There exists an algorithm to solve the APSP-AF problem on directed graphs with positive integer edge costs in $\tilde{O}(\sqrt{t}c^{(\omega+5)/4}n^{(\omega+9)/4})$ worst case time complexity.
\end{theorem}
\begin{proof}
Clearly we can take a similar approach to the method described in the proof of Theorem \ref{the:apsdaf}. We can still use the path cost (distance) to look up each successor node in $O(1)$ time. We note that the witnesses in the acceleration phase can be artificial vertices, but the corresponding real vertices can be retrieved in $O(1)$ time simply by storing this information when $G$ is expanded to $G'$.\qed
\end{proof}

If $c=1$, $\tilde{O}(\sqrt{t}c^{(\omega+5)/4}n^{(\omega+9)/4})$ becomes $\tilde{O}(\sqrt{t}n^{(\omega+9)/4})$, hence we have successfully generalized the APSP-AF problem from graphs with unit edge costs to graphs with integer edge costs. To compare with the straightforward method of repeatedly solving the APSP problem for each maximal flow value using Zwick's algorithm, we use the formula $\omega(1,r,1) = 2 + (\omega - 2)(r - \alpha)/(1-\alpha)$ where $O(n^{\omega(1,r,1)})$ is the time taken to multiply an $n$-by-$n^{r}$ matrix with an $n^{r}$-by-$n$ matrix, and $\alpha$ is a constant such that multiplying an $n$-by-$n^{\alpha}$ matrix with an $n^{\alpha}$-by-$n$ remains within the $O(n^{2})$ time bound. We let $\omega = 2.376$ and $\alpha = 0.294$ in this comparison \cite{Zwick02}. The time complexity of the straightforward method becomes $\tilde{O}(tn^{2+\mu})$, where $c = n^{x}$ such that the equation $\omega(1,\mu,1) = 1 + 2\mu - x$ is satisfied. With $t = O(n^{2})$, Algorithm \ref{alg:apsdaf2} is faster for $c < n^{0.629}$.

Finally we make a note that the idea of expanding the graph to $G'$ for the acceleration phase then contracting it back to $G$ for the cruising phase can retrospectively be applied to Algorithm \ref{alg:agm} to give a sharper time bound than $O((cn)^{(3+\omega)/2})$, which is the time bound given by Alon \emph{et al.} in their original paper \cite{AGM}. The time bound of $O((cn)^{(3+\omega)/2})$ is sub-cubic for $c < n^{0.117}$. Using our new approach of contracting the graph after the acceleration phase, the time bound can be improved to $O(c^{(1+\omega)/2}n^{(3+\omega)/2})$, which is sub-cubic for $c < n^{0.186}$. For solving the same problem as Algorithm \ref{alg:agm}, however, other algorithms are already known that remain sub-cubic for larger values of $c$ \cite{Tad95,Zwick02}.

\subsection{Concluding remarks}
The key achievements of this paper are: 1) the introduction of a new problem that clearly has numerous practical applications in network analysis involving both path costs and capacities, 2) non-trivial extension of Algorithm \ref{alg:agm} to solve the new problem that is more complex than the APSP problem, and 3) a better method to utilize the artificial graph for integer edge costs resulting in an improved time bound for not only our new algorithm, but also an existing algorithm for solving the APSP problem.

Solving the new APSP-AF problem on other types of graphs (e.g. undirected, real edge costs, etc) as well as finding efficient algorithms for the single source version of the problem remain on the agenda for future research.


\begin{thebibliography}{}

\bibitem{AHU}
A. V. Aho, J. E. Hopcroft, and J. D. Ullman:
The Design and Analysis of Computer Algorithms.
Addison-Wesley (1974)

\bibitem{AGM}
N. Alon, Z. Galil and O. Margalit:
On the Exponent of the All Pairs Shortest Path Problem.
Proc. $32^{nd}$ FOCS (1991) pp. 569--575

\bibitem{Chan}
T. Chan:
More algorithms for all-pairs shortest paths in weighted graphs.
Proc. $39^{th}$ STOC (2007) pp. 590--598


\bibitem{Dobosiewicz}
W. Dobosiewicz:
A more efficient algorithm for the min-plus multiplication.
International Journal of Computer Mathematics 32 (1990) pp. 49--60

\bibitem{DP}
R. Duan and S. Pettie:
Fast Algorithms for (max,min)-matrix multiplication and bottleneck shortest paths.
Proc. $19th$ SODA (2009) pp. 384--391

\bibitem{Floyd}
R. Floyd:
Algorithm 97: Shortest Path.
Communications of the ACM 5 (1962), pp. 345

\bibitem{Fredman}
M. Fredman:
New bounds on the complexity of the shortest path problem.
SIAM Journal on Computing 5 (1976), pp. 83--89


\bibitem{Gall}
F. Le Gall:
Faster Algorithms for Rectangular Matrix Multiplication.
Proc. $53^{rd}$ FOCS (2012) pp. 514--523

\bibitem{Han}
Y. Han:
An $O(n^{3}(\log{\log{n}}/\log{n})^{5/4})$ time algorithm for all pairs shortest paths.
Proc. $14^{th}$ ESA (2006), pp. 411--417

\bibitem{HT}
Y. Han and T. Takaoka:
An $O(n^{3}\log{\log{n}}/\log^{2}{n})$ Time Algorithm for All Pairs Shortest Paths.
Proc. $13^{th}$ SWAT (2012), pp. 131--141

\bibitem{Strassen}
A. Sch\"{o}nhage and V. Strassen:
Schnelle Multiplikation Gro$\beta$er Zahlen.
Computing 7 (1971) pp. 281--292

\bibitem{Tad95}
T. Takaoka:
Sub-cubic Cost Algorithms for the All Pairs Shortest Path Problem.
Algorithmica 20 (1995) pp. 309--318

\bibitem{Tad04}
T. Takaoka
A faster algorithm for the all-pairs shortest path problem and its application.
Proc. $10^{th}$ COCOON (2004) pp. 278--289

\bibitem{VWY}
V. Vassilevska, R. Williams, R. Yuster:
All Pairs Bottleneck Paths and Max-Min Matrix Products in Truly Subcubic Time.
Journal of Theory of Computing 5 (2009) pp. 173--189

\bibitem{Williams}
V. Williams:
Breaking the Coppersmith-Winograd barrier.
Proc. $44^{th}$ STOC (2012)

\bibitem{Zwick02}
U. Zwick:
All Pairs Shortest Paths using Bridging Sets and Rectangular Matrix Multiplication.
Journal of the ACM 49 (2002) pp. 289--317

\bibitem{Zwick06}
U. Zwick:
A Slightly Improved Sub-Cubic Algorithm for the All Pairs Shortest Paths Problem with Real Edge Lengths.
Algorithmica 46 (2006) pp. 278--289


\end{thebibliography}
\end{document}